\documentclass[a4paper,10pt]{article} \usepackage{rpicstyle}
\usepackage[vlined,algoruled]{algorithm2e}
\usepackage{amsmath,amssymb,amsthm}
\usepackage{enumerate,xcolor}
\usepackage{url}
\newcommand{\R}{\mathbb{R}}
\newcommand{\C}{\mathbb{C}}

\newcommand{\X}{\mathcal{X}}

\newcommand{\Z}{\mathcal{Z}}

\newcommand{\Sb}{\mathcal{S}}
\newcommand{\Tb}{\mathcal{T}}

\newcommand{\B}{\mathcal{B}}

\newcommand{\rank}{\mathrm{rank}}

\newcommand{\spn}{\mathrm{Span}}

\DeclareMathOperator{\img}{\mathrm{img}}

\newcommand{\diag}{\mathrm{diag}}
\newcommand{\srad}{\boldsymbol{\rho}}
\newcommand{\id}{\mathrm{I}}
\newcommand{\s}{N}
\newcommand{\CL}{\textsc{cl}}

\newcommand{\dfn}{\doteq}

\newcommand{\Ind}{\underline{\s}}%{\mathcal{I}}
\newcommand{\cea}{CEA\xspace}

\newcommand{\blk}{\mathrm{blk}}

\newcommand{\QEDopen}{\hfill$\Box$}
\renewcommand{\dim}{\mathrm{d}}

\newtheorem{Lemma}{Lemma}
\newtheorem{Theorem}{Theorem}

\newtheorem{Definition}{Definition}

\newtheorem{Corollary}{Corollary}

\begin{document}

\title{Sufficient conditions for the genericity of feedback
  stabilisability of switching systems via Lie-algebraic solvability}
\author{H. Haimovich$^{\dag }$ \and J. H. Braslavsky$^{\ddag }$}
\affiliation{%
  \dag CONICET and Laboratorio de Sistemas Din\'amicos y
  Procesamiento de Informaci\'on,\\
  Depto. de Control, Esc. de Ing. Electr\'onica, FCEIA,
  Universidad Nacional de Rosario,\\
  Riobamba 245bis, 2000 Rosario, Argentina. {\ h.haimovich@gmail.com} \\
  \ddag Australian Commonwealth Scientific and Industrial Research
  Organisation (CSIRO),\\
  Division of Energy Technology, PO Box 330, Newcastle NSW 2300,
  Australia. {\ Julio.Braslavsky@csiro.au} }

\maketitle

\abstract %
This paper addresses the stabilisation of discrete-time switching
linear systems (DTSSs) with control inputs under arbitrary switching,
based on the existence of a common quadratic Lyapunov function (CQLF).
The authors have begun a line of work dealing with control design
based on the Lie-algebraic solvability property. The present paper
expands on earlier work by deriving sufficient conditions under which
the closed-loop system can be caused to satisfy the Lie-algebraic
solvability property generically, i.e. for almost every set of system
parameters, furthermore admitting straightforward and efficient
numerical implementation. %
\endabstract


\section{INTRODUCTION}

This paper considers control design for feedback stabilisation of
DTSSs under arbitrary switching regimes. Several results for the
analysis of stability of switching systems under arbitrary switching
exist \cite{showir_siamrev07,lin09:_stabil}, although most of these
deal only with autonomous switching systems, i.e. switching systems
without continuous control inputs.

The study of control design methods to achieve closed-loop stability
under arbitrary switching for systems with control inputs, in
contrast, has been relatively scarce. Amidst the existing work, a
computationally appealing approach consists in adapting LMI-based
numerical stability tests to the synthesis of feedback controls that
will guarantee the existence of a CQLF for the closed-loop switching
system (e.g. \cite{daafouz02:_stabil,ji04:_LMI}). Such
LMI-based methods for control design, however, provide little
information on the structure of the designed closed-loop system, and
thus in general offer limited support for analysis.

% Research on these systems has recently led to significant
% advances on the stability theory of switching and hybrid systems
% \cite{showir_siamrev07,lin09:_stabil}. Many of these advances have
% focused on autonomous switching linear systems, that is, systems
% without continuous control inputs.

% In an effort to provide analytic tools supporting control design for
% these systems, the authors of \cite{wulff09:_contr_dsgn} propose a
% pole-placement approach based on the design of controllers targeting a
% closed-loop common eigenvector structure that guarantees global
% exponential stability under arbitrary switching.  At the expense of
% generality, the method proposed in \cite{wulff09:_contr_dsgn} provides
% analytically transparent solutions for a practically important class
% of systems.

% With a similar motivation to that in \cite{wulff09:_contr_dsgn}, t

A structurally-based control design method has been developed by the
authors \cite{haimovich11:_lie,haibra_cdc10,haibra_aucc11} that
involves the Lie-algebraic solvability property. A well-known
sufficient condition for stability \cite{liberzon03:_switc} states
that an autonomous switching linear system admits a CQLF (and is hence
stable under arbitrary switching) if every subsystem is stable and the
Lie algebra generated by the subsystem matrices is
solvable. Solvability of a matrix Lie algebra is equivalent to the
existence of a single similarity transformation that transforms each
matrix into upper triangular form. The results of
\cite{haimovich11:_lie,haibra_cdc10,haibra_aucc11} thus ``activate''
the aforementioned stability analysis result into a control design
technique.
% In \cite{haimovich11:_lie}, such control design technique
% consists in an iterative control design algorithm.

% technique to achieve closed-loop feedback structure that will satisfy
% (and in \cite{haibra_cdc10} \emph{approximately satisfy}) the
% aforementioned Lie-algebraic-solvability condition is developed.

% The present paper contributes to the analysis of discrete-time
% switching systems (DTSS) with control inputs, building on work
% initiated in \cite{haimovich11:_lie} and continued in
% \cite{haibra_cdc10,haibra_aucc11}. In \cite{haimovich11:_lie}, an
% iterative control design technique to achieve a closed-loop feedback
% structure that will satisfy (and in \cite{haibra_cdc10}
% \emph{approximately satisfy}) Lie-algebraic-solvability conditions
% that guarantee the existence of a CQLF is developed.
% %   , and thus the
% % uniform (with respect to all switching sequences) global exponential
% % stability of the system \cite{liberzon03:_switc}.
% These
% Lie-algebraic-solvability conditions are equivalent to the existence
% of a common similarity transformation that simultaneously renders each
% of the stable closed-loop subsystem matrices upper triangular (or
% nearly upper triangular in \cite{haibra_cdc10}). While the design
% technique proposed in \cite{haimovich11:_lie,haibra_cdc10} is
% restrictive, in that the target closed-loop structure is only a
% sufficient condition for the existence of a CQLF, the approach offers
% insight into the fundamental structure and controllability properties
% of DTSS.% These properties and structure are further
% % expounded in the present paper.

A central contribution in \cite{haimovich11:_lie} is an iterative
design algorithm that searches for a set of stabilising feedback
matrices that attain the target simultaneously triangularisable
closed-loop structure via the application of a common eigenvector
assignment (CEA) procedure and state dimension reduction at each
iteration. % This CEA procedure searches for a vector and feedback
% matrices that make such a vector an eigenvector common to all the
% closed-loop subsystems internal to the algorithm, having stable
% eigenvalues.
The main theoretical result in \cite{haimovich11:_lie} establishes
that the proposed algorithm will be successful until the state
dimension is reduced to 1 \emph{if and only if} feedback matrices
exist so that the corresponding closed-loop subsystem matrices are
stable and simultaneously triangularisable, i.e. if and only if
feedback matrices exist so that the closed-loop system satisfies the
aforementioned Lie-algebraic stability condition. Also in
\cite{haimovich11:_lie}, a numerical implementation for the proposed
iterative design algorithm and CEA procedure are provided. A key
structural condition also is provided which, when satisfied, guarantees
a directly computable solution for the CEA procedure. If this
structural condition is not satisfied, then the required quantities
are sought by means of an optimisation problem.

The aforementioned Lie-algebraic stability condition is (a)
restrictive and (b) non-robust, in the sense that (a) it is satisfied
for a very limited number of autonomous switching systems and (b) even
if it is satisfied for a given system, it is almost surely not
satisfied by systems with parameters arbitrarily close to the given
one. The work in \cite{haibra_cdc10} then provides a robust result by relaxing, for single
input systems, % the iterative design algorithm and theoretical results
% in \cite{haimovich11:_lie} to the case in which the target closed-loop
% structure is relaxed
the simultaneous triangularisation requirement to \emph{approximate}
(in a specific sense) simultaneous triangularisation. The main
theoretical contribution in \cite{haibra_cdc10} establishes that if a
system satisfying the aforementioned Lie-algebraic condition exists in
a suitably small neighbourhood of the given system data, then the
proposed algorithm is guaranteed to find feedback matrices so that the
corresponding closed-loop DTSS admits a CQLF even if the Lie-algebraic
condition is not met by the given system data.

Even if the aforementioned Lie-algebraic condition is restrictive for
autonomous switching systems, the existence of feedback controls
causing the corresponding closed-loop system to satisfy such
Lie-algebraic condition need not be such a restrictive condition. The
restrictiveness of this condition is related to the key structural
condition provided in \cite{haimovich11:_lie}: if such structural
condition is satisfied at every iteration of the algorithm, then the
problem may be not restrictive at all for systems with the given
dimensions. In this regard, the main result in \cite{haibra_aucc11} is
the identification of the situation that prevents the structural
condition from holding at every iteration of the algorithm.

In the present paper, we build upon the results of
\cite{haibra_aucc11} by providing sufficient conditions for the
structural condition to hold at every iteration of the algorithm
\emph{for almost every set of system parameters with the given
  dimensions}. We thus provide sufficient conditions for the
\emph{genericity} of the property of existence of feedback matrices so
that the closed-loop subsystem matrices are stable and generate a
solvable Lie algebra.

\textbf{Notation}. The index set $\{1,2,\dots,\s\}$ is denoted
$\Ind$. The kernel (null space) of a matrix or linear map $A$ is
denoted $\ker A$, its image (range), $\img A$, and its spectral
radius, $\srad(A)$. For $x\in\C^{n\times m}$, its transpose is denoted
$x'$, its conjugate transpose $x^*$ and its Moore-Penrose generalised
inverse $x^\dagger$. If $\Sb,\Tb$ are vector spaces, then
$\Sb\subset\Tb$ means that $\Sb$ is a subspace of $\Tb$ and
$\dim(\Sb)$ denotes the dimension of $\Sb$.

%%%%%%%%%%%%%%%%%%%%%%%%%%%%%%%%%%%%%%%%%%%%%%%%%%%%%%%%%%%%%%%%%%%%%%%%%%%%%%%%
\section{PROBLEM FORMULATION}
\label{sec:problem-formulation}

Consider the DTSS
\begin{eqnarray}
  \label{eq:dtss}
  x_{k+1}&=& A_{i(k)}x_k + B_{i(k)}u_k^{i(k)},
\end{eqnarray}
where $x_k \in \R^n$ and $u_k^{i} \in \R^{m_i}$ for all $k$, $i(k)$
takes values in $\Ind$ for all $k$, the matrices $A_i \in \R^{n\times
  n}$ and $B_i \in \R^{n\times m_i}$ are known for all $i\in\Ind$,
$B_i$ have full column rank, and $(A_i,B_i)$ is controllable for all
$i\in\Ind$. We are interested in state-feedback control design of the
form
\begin{equation}
  \label{eq:47}
  u_k^{i(k)} = K_{i(k)}x_k,
\end{equation}
so that the resulting closed-loop system
\begin{align}
  \label{eq:26}
  x_{k+1} &=  A_{i(k)}^\CL x_k,\quad\text{where}\\
  \label{eq:49}
  A_{i}^\CL &= A_{i} + B_{i} K_{i},\quad\text{for }i\in\Ind,
\end{align}
admit a CQLF and hence be stable under arbitrary switching. Note that
at every time instant $k$, the control law \eqref{eq:47} requires
knowledge of the ``active'' subsystem given by $i(k)$.

As is well-known, ensuring that $\srad(A_i^\CL)<1$ for $i\in\Ind$
is necessary but not sufficient to ensure the stability of the DTSS
(\ref{eq:26}) for arbitrary switching. A sufficient condition is
given by the following result, which is a minor modification of
\cite[Theorem 6.18]{theys_phd05}.

\begin{Lemma}[Lie-algebraic-solvability stability condition]
  \label{lem:stab}
  If $\srad(A_i^\CL)<1$ for $i\in\Ind$, and the Lie algebra
  generated by $\{A_i^\CL: i\in\Ind\}$ is solvable, then
  (\ref{eq:26}) admits a common quadratic Lyapunov function and hence
  is exponentially stable.\hfill\QEDopen
\end{Lemma}

In this paper, we specifically consider stabilising state feedback
design for the DTSS \eqref{eq:dtss} based on the
Lie-algebraic-solvability condition of Lemma~\ref{lem:stab}, and thus
focus on the DTSS class defined next.

\begin{Definition}[SLASF]% $\gamma$-SLASF]
  \label{def:SLASF}
  A set $\Z = \{(A_i \in \R^{n\times n}, B_i \in \R^{n\times m_i}) :
  i\in\Ind\}$ is said to be \emph{SLASF (Solvable Lie Algebra with
    Stability by Feedback)} if there exist $K_i \in \R^{m_i\times n}$
  such that $A_i^\CL$ as in (\ref{eq:49}) generate a solvable Lie
  algebra and satisfy $\srad(A_i^\CL)<1$. % If $\Z$ is
  % SLASF, we say that $K_i \in \C^{m_i\times n}$ are \emph{compatible
  %   with $\Z$} if $A_i^\CL$ as in (\ref{eq:49}) generate a solvable
  % Lie algebra and satisfy $\srad(A_i^\CL)<1$. 
  \hfill\QEDopen 
\end{Definition}

In matrix terms, the fact that the Lie algebra generated by the
matrices $A_i^\CL$ is solvable is equivalent to the existence of an
invertible matrix $T\in\C^{n\times n}$ such that $T^{-1} A_i^\CL T$ is
upper triangular for $i\in\Ind$. That is, each matrix $A_i^\CL$ is
similar to an upper triangular matrix under \emph{a common} similarity
transformation $T$.  Note that even if the matrices $A_i^\CL$ have
real entries, those of $T$ may be complex \cite{erdwil_book06}.

\section{PREVIOUS RESULTS}
\label{sec:contr-design-algor}

Control design that causes the closed-loop system to be stable by
satisfying the conditions of Lemma~\ref{lem:stab} can be performed
iteratively by seeking feedback matrices that assign a common
eigenvector with stable corresponding eigenvalues, and reducing the
state-space dimension by 1 at every iteration
\cite{haibra_cdc09,haimovich11:_lie}. This methodology is given in
pseudocode below as Algorithm~\ref{alg:main}. Algorithm~\ref{alg:main}
seeks feedback matrices $K_i$ so that the closed-loop matrices
$A_i^\CL$ given by (\ref{eq:49}) are stable and simultaneously
triangularisable.

% The core of Algorithm~\ref{alg:main} is Procedure CEA in
% (\ref{eq:15}), which at every iteration [$\ell$ indicates iteration
% number, see (\ref{eq:5})] aims to compute a vector, $v_1^\ell$, and
% corresponding feedback matrices, $\{F_i^\ell\}_{i=1}^N$, so that
% $v_1^\ell$ is a feedback-assignable eigenvector common to all
% subsystems, with corresponding stable eigenvalues.
% % For simplicity of presentation, we will only deal with
% % real eigenvalues throughout; for complex eigenvalues, both
% % Algorithm~\ref{alg:main} and Procedure~\cea require modifications to
% % guarantee the assignment of corresponding conjugate eigenvectors, and
% % the computation of real feedback matrices.  Such modifications have
% % been numerically implemented in \cite{haibra_cdc10}.

% The following section first describes in more detail the proposed
% control design methodology (Algorithm~\ref{alg:main}), and then
% reviews the main supporting theoretical result, adapted from
% \cite{haimovich11:_lie,haibra_cdc10}.  The main results of the present
% paper deal with sufficient conditions for the generic existence of
% stabilising solutions to the triangularisation by feedback
% performed by Algorithm~\ref{alg:main}, and are given in
% Section~\ref{sec:contr-design-algor-2}.

\begin{algorithm}[!ht]
  \label{alg:main}
  \SetTitleSty{texrm}{\normalsize}
  \SetKwComment{tcc}{\% }{}
  \KwData{$A_i \in \R^{n\times n}$, $B_i \in \R^{n\times m_i}$ for
    $i\in\Ind$}
  \KwOut{$K_i$ for $i\in\Ind$}
  \SetKwBlock{Init}{begin}{end}
  \Init(Initialisation){
    $A_i^1 \dfn A_i$, $B_i^1 \dfn B_i$, $K_i^0 \dfn 0$, $U_1 \dfn \id$, $\ell \leftarrow 0$ \;
%    $U \dfn [\,]$ (empty), 
  }
  \Repeat{$\ell = n$}%
  {\vspace{-1em}
    \begin{align}
      \label{eq:5}
      \ell &\leftarrow \ell+1,\quad n_\ell \leftarrow n-\ell+1 \;\\
      \label{eq:15}
      [v_1^\ell ,\{&F_i^\ell\}_{i=1}^{\s}] 
      \leftarrow  \text{\cea}(\{A_i^\ell\}_{i=1}^{\s},
      \{B_i^\ell\}_{i=1}^{\s})\\
      \label{eq:33}
      A_i^{\ell,\CL} &\dfn A_i^\ell + B_i^\ell F_i^\ell,\\
      % \notag
      % (U)_{:,\ell} &\leftarrow \biggl(\prod_{r=\ell}^{r=1} U_r\biggr)
      % v_1^\ell = U_{1} U_{2} \cdots U_{\ell} v_1^\ell,\\
      \label{eq:9}
      K_i^\ell &\leftarrow K_i^{\ell-1} + F_i^\ell 
      \biggl(\prod_{r=1}^{\ell} U_r^* \biggr)
    \end{align}
    \If{$\ell<n$}{%
      Construct a unitary matrix:
      \begin{equation}
        \label{eq:62}
        \big[v_1^\ell | v_2^\ell | \cdots 
        | v_{n_\ell}^\ell\big] \in \C^{n_\ell \times n_\ell}.
      \end{equation}
      Assign\vspace{-2\baselineskip}
      \begin{align}
        \label{eq:99}
        U_{\ell+1} &\leftarrow [v_2^{\ell}|\cdots |v_{n_\ell}^\ell], \\
        \label{eq:100}
        A_i^{\ell+1} &\leftarrow U_{\ell+1}^{*} A_i^{\ell,\CL} U_{\ell+1},\\
        \label{eq:6} 
        B_i^{\ell+1} &\leftarrow U_{\ell+1}^{*} B_i^\ell,
      \end{align}
    }
  }
  $K_i \leftarrow K_i^n$ \;
  \caption{Iterative triangularisation}
\end{algorithm}

\subsection{The Algorithm}
\label{sec:algorithm}

Algorithm~\ref{alg:main} begins by setting internal matrices equal to
the subsystem matrices of the DTSS to be stabilised ($A_i^1 = A_i$ and
$B_i^1 = B_i$ at the Initialisation step). At every iteration [$\ell$
indicates iteration number, see (\ref{eq:5})], the algorithm executes
Procedure~\cea [see~(\ref{eq:15})] on its internal system matrices
(the latter matrices are $A_i^\ell$ and $B_i^\ell$). Procedure~\cea
aims to compute a vector, $v_1^\ell$, and corresponding feedback
matrices, $F_i^\ell$, so that $v_1^\ell$ is a feedback-assignable unit
eigenvector common to all internal subsystems, with corresponding
stable eigenvalues. That is, if Procedure~\cea is successful, then
$v_1^\ell$ will satisfy $\|v_1^\ell\| = 1$ and $(A_i^\ell + B_i^\ell
F_i^\ell) v_1^\ell = \lambda_i^\ell v_1^\ell$ for some scalars
$\lambda_i^\ell$ satisfying $|\lambda_i^\ell| < 1$, for all
$i\in\Ind$. Algorithm~\ref{alg:main} then computes internal
closed-loop matrices [$A_i^{\ell,\CL}$ in (\ref{eq:33})], updates
internal feedback matrices [$K_i^\ell$ in (\ref{eq:9})] and then
reduces the internal state dimension by 1. This reduction occurs at
(\ref{eq:62})--(\ref{eq:6}) [$n_\ell$ is the internal state dimension,
see (\ref{eq:5})]. Note that $v_1^\ell$ is the first column of the
unitary matrix (\ref{eq:62}), and considering (\ref{eq:99}) then
$U_{\ell+1}^* U_{\ell+1} = \id$ and $U_{\ell+1}^* v_1^\ell =
0$. Algorithm~\ref{alg:main} iterates until the internal state reaches
dimension 1. If the given system matrices form a SLASF set (recall
Definition~\ref{def:SLASF}), the matrices $K_i$ computed by
Algorithm~\ref{alg:main} will be the required feedback matrices.

If the given system matrices $A_i$, $B_i$, for $i\in\Ind$, form a
SLASF set, then at every iteration of Algorithm~\ref{alg:main} a
stable feedback-assignable common eigenvector $v_1^\ell$ is ensured to
exist for the internal system with matrices $A_i^\ell$, $B_i^\ell$,
for $i\in\Ind$. Conversely, if a feedback-assignable common
eigenvector $v_1^\ell$ exists at every iteration of
Algorithm~\ref{alg:main}, then the given system matrices form a SLASF
set. The latter constitutes the main theoretical result that underpins
our iterative control design algorithm
\cite{haibra_cdc09,haimovich11:_lie}. % One additional important
% property is that if $(A_i^\ell,B_i^\ell)$ is controllable and
% $v_1^\ell$ is a feedback-assignable common eigenvector, then
% $(A_i^{\ell+1},B_i^{\ell+1})$ is controllable (see, e.g. Proposition
% 1.2 of \cite{wonham_book85}). 

\subsection{The Procedure}
\label{sec:keystruc}

As expressed in the previous paragraph, the existence of a
feedback-assignable common eigenvector with corresponding stable
eigenvalues is central to our development. This section recalls the
structural condition introduced in \cite{haimovich11:_lie} which, when
satisfied, ensures that such a vector exists and allows its
computation in a numerically efficient and straightforward way.
% characterisation of when this property will hold with genericity at
% iteration $\ell$, i.e. for almost every set of entries of $A_i^\ell$
% and $B_i^\ell$.
% according to a structural condition that depends on the
% ranks and dimensions of the matrices that define the system. 

We introduce some notation required to state the aforementioned
structural condition. Define $m_i^\ell \dfn \rank(B_i^\ell) =
\dim(\img B_i^\ell)$, and
factor $B_i^\ell = b_i^\ell r_i^\ell$, where $r_i^\ell : \R^{m_i} \to
\R^{m_i^\ell}$ has full row rank and $b_i^\ell : \R^{m_i^\ell} \to
\R^{n_\ell}$ has full column rank. We adopt the convention that
$b_i^\ell$ is an empty matrix if $m_i^\ell=0$. Note that $\img
B_i^\ell = \img b_i^\ell$. Let $\Lambda^\ell$ be the vector with
components $\lambda_i^\ell$, $i\in\Ind$, i.e.
\begin{equation}
  \label{eq:31}
  \Lambda^\ell \dfn [\lambda_1^\ell, \lambda_2^\ell, \ldots, \lambda_N^\ell]',
\end{equation}
and build the matrix
\begin{align}
  \label{eq:17}
  Q_\ell(\Lambda^\ell) &\dfn [R_\ell(\Lambda^\ell) , -B_\ell]\\
  R_\ell(\Lambda^\ell) &\dfn
  \begin{bmatrix}
    \lambda_1^\ell I - A_1^\ell\\
    \vdots \\
    \lambda_N^\ell I - A_N^\ell
  \end{bmatrix},
%(\lambda_2 I -A_2^\ell)',
  &B_\ell &\dfn \blk\diag\bigl[b_1^\ell,\dots,b_N^\ell\bigr],\notag
\end{align}
where $\blk\diag$ denotes block diagonal concatenation.

\begin{Lemma}[Structural condition \cite{haimovich11:_lie,haibra_aucc11}]
  \label{lem:struct}
  Let
  \begin{equation}
    \label{eq:8}
    p_\ell \dfn n_\ell + \sum_{i=1}^N m_i^\ell - N n_\ell.
  \end{equation}
  Then,
  \begin{enumerate}[(a)]
  \item A vector that can be assigned by feedback as a common
    eigenvector with corresponding eigenvalues $\lambda_i^\ell$ for
    $i\in\Ind$ exists if and only if $\dim(\ker Q_\ell(\Lambda^\ell))
    > 0$% , with $\Lambda^\ell$ as in (\ref{eq:31})
    .\label{item:1}
  \item If $Q_\ell(\Lambda^\ell)w = 0$ with $w\neq 0$ partitioned as
    \begin{gather}
      \label{eq:13}
      w \dfn [v',u_1',\dots,u_N']',\qquad\text{then }v\neq 0,\text{ and}\\
      \label{eq:39}
      (A_i^\ell +B_i^\ell F_i^\ell)v = \lambda_i^\ell v,\quad\text{for }i\in\Ind,
    \end{gather}
    for every $F_i^\ell$ satisfying $r_i^\ell F_i^\ell v = u_i$. For
    each $i\in\Ind$ one such $F_i^\ell$ always exists and is given by
    $F_i^\ell = (r_i^\ell)^\dagger u_i v^\dagger$.\label{item:4}
  \item $\dim(\ker Q_\ell(\Lambda^\ell)) \ge p_\ell$ for every choice
    of $\Lambda^\ell$ as in (\ref{eq:31}). Consequently, if $p_\ell >
    0$, then a feedback-assignable common eigenvector exists for every
    choice of corresponding eigenvalues.\label{item:3}
    \mbox{}\hfill\QEDopen
  \end{enumerate}
\end{Lemma}
Lemma~\ref{lem:struct} gives a structural condition, namely $p_\ell >
0$, for a feedback-assignable common eigenvector $v$ %in \eqref{eq:39}
to exist for every choice of corresponding eigenvalues
$\lambda_i^\ell$. This condition is \emph{structural} because the
quantities involved in the computation of $p_\ell$ are only matrix
ranks and dimensions. If the structural condition $p_\ell > 0$ is
satisfied, a feedback-assignable common eigenvector $v_1^\ell$, as
required at iteration $\ell$ of Algorithm~\ref{alg:main}, can be
computed by (a) selecting its closed-loop eigenvalues $\lambda_i^\ell$
corresponding to each subsystem, (b) finding a vector $w \neq 0$ with
components partitioned as in (\ref{eq:13}) so that
$Q_\ell(\Lambda^\ell)w = 0$, i.e. so that $w \in \ker
Q_\ell(\Lambda^\ell)$, (c) taking the first $n_\ell$ components of
$w$, i.e. the subvector $v$ in (\ref{eq:13}), and (d) computing
$v_1^\ell = v/\|v\|$. The feedback matrices that assign such
eigenvector with corresponding eigenvalues $\lambda_i^\ell$ can be
obtained as $F_i^\ell = (r_i^\ell)^\dagger u_i v^\dagger$.

An implementation of Procedure~\cea is thus given below for the case
when the structural condition of Lemma~\ref{lem:struct} is satisfied.
\begin{procedure}[htbp]
  \SetTitleSty{texrm}{\normalsize} \SetKwComment{tcc}{\% }{}
%  \LinesNumberedHidden
  \KwIn{$A_i^\ell \in \R^{n_\ell \times n_\ell}$, $B_i^\ell \in
    \R^{n_\ell \times m_i}$, for $i\in\Ind$}
  \KwOut{$v_1^\ell$, $F_i^\ell$ for $i\in\Ind$} %
  Factor $B_i^\ell = b_i^\ell r_i^\ell$ with
  $b_i^\ell\in\R^{n_\ell\times m_i^\ell}$ and $m_i^\ell =
  \rank(B_i^\ell)$ \;
%  for $i\in\Ind$ \; %
  \If{$p_\ell = n_\ell + \sum_{i=1}^N m_i^\ell - N n_\ell>0$} %
%  {\tcc{\small Structural condition satisfied} %
    {Select $\lambda_i^\ell \in \R$ so that $|\lambda_i^\ell|<1$ \;
    Find $w\neq 0$ such that $Q_\ell(\Lambda^\ell)w = 0$ \; %
    Partition $w$ as in (\ref{eq:13}) \; %
    $v_1^\ell = v/\|v\|$ \;
    $F_i^\ell = (r_i^\ell)^\dagger u_i\, v^\dagger$, for $i\in\Ind$ \;}%
  % \Else{%
  %   \tcc{\small Optimisation required
  %   (see \cite{haimovich11:_lie,haibra_cdc10})}}
  \caption{CEA () (Structural condition satisfied)}
  \label{proc:aceas}
\end{procedure}

Even if the DTSS matrices $A_i$, $B_i$ have real entries, those of the
matrices $A_i^\ell$, $B_i^\ell$ internal to Algorithm~\ref{alg:main}
can be complex at some iteration $\ell$. This is so because the vector
$v_1^\ell$ returned by Procedure~\cea (a feedback-assignable common
eigenvector) can have complex components even if $A_i^\ell$,
$B_i^\ell$ have real entries, causing $A_i^{\ell+1}$, $B_i^{\ell+1}$
to have complex entries. However, when the structural condition
$p_\ell > 0$ is satisfied, the closed-loop eigenvalues $\Lambda^\ell$
can be arbitrarily selected. Hence, selecting real closed-loop
eigenvalues will cause the vector $v_1^\ell$ to have real
components. In the sequel, we assume that real eigenvalues will be
selected and hence all matrices internal to Algorithm~\ref{alg:main}
will have real entries.

\subsection{The Structural Condition}
\label{sec:structural-condition}

If the structural condition given by Lemma~\ref{lem:struct}, namely
$p_\ell > 0$, holds at iteration $\ell$ of Algorithm~\ref{alg:main},
then Procedure~\cea can easily compute a feedback-assignable common
eigenvector and the corresponding feedback matrices, for every choice
of corresponding closed-loop eigenvalues. % We next aim at deriving
% conditions to ensure that $p_\ell > 0$ for $\ell=1,\ldots,n$. 
The quantity $p_\ell$ depends on $m_i^\ell$, the rank of
$B_i^\ell$. At the first iteration of Algorithm~\ref{alg:main},
i.e. when $\ell=1$, the internal matrices $B_i^1 = B_i$ have $n=n_1$
rows, $m_i$ columns, and since by assumption they have full column
rank, then $m_i^1 = m_i$. At subsequent iterations, the matrices
$B_i^{\ell}$ have $n_\ell = n - \ell + 1$ rows and $m_i$ columns.
According to (\ref{eq:62})--(\ref{eq:99}) and (\ref{eq:6}), we have
\begin{equation}
  \label{eq:2}
  m_i^\ell - 1 \le m_i^{\ell+1} \le m_i^\ell.
\end{equation}
According to (\ref{eq:62})--(\ref{eq:99}) and (\ref{eq:6}),
$m_i^{\ell+1}$ will depend on the feedback-assignable eigenvector
$v_1^\ell$ returned by Procedure~\cea:
\begin{equation}
  \label{eq:1}
  m_i^{\ell+1} =
  \begin{cases}
    m_i^\ell &\text{if }v_1^\ell \notin \img B_i^\ell,\\
    m_i^\ell - 1 &\text{if }v_1^\ell \in \img B_i^\ell.\\
  \end{cases}
\end{equation}
From (\ref{eq:1}), then $m_i^{\ell+1} = m_i^\ell - 1$ when
$m_i^\ell = n_\ell$, because necessarily in this case $v_1^\ell \in
\R^{n_\ell} = \img B_i^\ell$. The following theorem and its corollary
follow from (\ref{eq:1}) and were presented in \cite{haibra_aucc11}.
\begin{Theorem}
  \label{thm:main-results}
  Consider Algorithm~\ref{alg:main} at iteration $\ell$ and $p_\ell$
  as in (\ref{eq:8}), with $m_i^\ell = \rank(B_i^\ell)$. Then,
  $p_{\ell+1}\ge p_\ell -1$, with equality if and only if 
  \begin{equation}
    \label{eq:35}
    v_1^\ell \in \bigcap_{i\in\Ind} \img B_i^\ell.% \quad \text{and $p_\ell = 1$.}
  \end{equation}
\end{Theorem}
\begin{Corollary}
\label{cor:fail-gCEAS}
  Let $p_\ell > 0$. Then,
  \begin{enumerate}[(a)]\itemsep-0pt
  \item $p_q > 0$ for $q=\ell,\ldots,\ell+p_\ell-1$.\label{item:5}
  \item $p_{\ell+1} > 0$ if $v_1^\ell \notin \img B_k^\ell$ for some $k\in\Ind$.\label{item:12}
  \item \label{item:9} $p_{\ell+1} \not> 0$ if and only if $p_\ell =
    1$ and (\ref{eq:35}).
  % \begin{equation}
  %   \label{eq:35}
  %   v_1^\ell \in \bigcap_{i\in\Ind} \img B_i^\ell, \quad \text{and $p_\ell = 1$.}
  % \end{equation}
  \end{enumerate}
\end{Corollary}

\section{MAIN RESULTS}
\label{sec:contr-eigvect}

In this section, we derive conditions to ensure that the structural
condition $p_\ell > 0$ will hold for $\ell=1,\ldots,n$. We will
achieve this goal by looking more deeply into the condition
(\ref{eq:35}). In Section~\ref{sec:maps-subspaces}, we recall a
property of subspaces that is required for the derivation of our main
results in Section~\ref{sec:induct-gener-struct}.

\subsection{Transversality of Subspaces}
\label{sec:maps-subspaces}

We next recall the property of transversality of subspaces (see, e.g.,
Chapter~0 of \cite{wonham_book85}). % This property will be employed
% in the derivation of our main results in
% Section~\ref{sec:induct-gener-struct}. 
\begin{Definition}[Transverse]
  Two subspaces $\Sb,\Tb$ of an ambient space $\X$ are said to be
  transverse when the dimension of their intersection is minimal,
  i.e. when
  \begin{equation}
    \label{eq:57}
    \dim(\Sb \cap \Tb) = \max\{0,\dim(\Sb) + \dim(\Tb) - \dim(\X)\}.
  \end{equation}
  Equivalently, $\Sb$ and $\Tb$ are transverse when the dimension of
  their sum is maximal. We extend this definition to sets of subspaces
  as follows. Let $S = \{\Sb_1,\ldots,\Sb_N\}$ be a set of subspaces
  of an ambient space $\X$. We say that $S$ is transverse when both
  the intersection of the subspaces in every subset of $S$ has minimal
  dimension and the sum of the subspaces in every subset of $S$ has
  maximal dimension.
\end{Definition}

The following properties of subspaces can be
straightforwardly established.
\begin{Lemma}
  \label{lem:transverse}
  Let $S = \{\Sb_1,\ldots,\Sb_N\}$ be a set of subspaces of the
  ambient space $\X$, and define
  \begin{equation*}
    p \dfn \dim(\X) + \sum_{i\in\Ind} \dim(\Sb_i) - N\dim(\X).
  \end{equation*}
  Then,
  \begin{enumerate}[(a)]\itemsep-0pt
  \item \label{item:2} $\dim(\Sb_i \cap \Sb_j) = \dim(\Sb_i) +
    \dim(\Sb_j) - \dim(\Sb_i+\Sb_j)$.
  \item \label{item:7} If $S$ is transverse, then
    $\dim\left(\bigcap_{i\in\Ind} \Sb_i \right) =
    \max \left\{0,p \right\}$.
  \item \label{item:8} If $S$ is transverse and $p \ge 0$, then
    $\dim(\Sb_i + \Sb_j) = \dim(\X)$ for all $i,j \in \Ind$ with
    $i\neq j$.
  \item \label{item:6} Let $J = I \cup \{j\}$, with $J\subset\Ind$ and
    $\# J = \# I + 1$. Suppose that $p\ge 0$ and that $\{\Sb_i : i\in
    I\}$ is transverse.
    % , and that $\Sb_i + \Sb_j = \X$ for all $i,j \in
    % J$ with $i\neq j$
    Then, $\{\Sb_i : i\in J\}$ is transverse if and only if
    $\bigcap_{i\in I} \Sb_i + \Sb_j = \X$.
  \end{enumerate}
\end{Lemma}
\begin{proof}[Proof of Lemma~\ref{lem:transverse}(\ref{item:6})]
  ($\Rightarrow$) For a set $K \subset \Ind$, define $p_K = \dim(\X) +
  \sum_{i\in K} \dim(\Sb_i) - \# K \dim(\X)$. Since $p_{\Ind} = p \ge
  0$, then $p_I \ge 0$ and $p_J \ge 0$ because $\dim(\Sb_i) \le
  \dim(\X)$ for all $i\in\Ind$. By
  Lemma~\ref{lem:transverse}(\ref{item:7}) and since $p_I\ge 0$ and
  $p_J\ge 0$, then $\dim(\bigcap_{i\in I} \Sb_i) = p_I$ and
  $\dim(\bigcap_{i\in J} \Sb_i) = p_J$. By
  Lemma~\ref{lem:transverse}(\ref{item:2}), we have
  \begin{align}
    \dim(\bigcap_{i\in J} \Sb_i) &= \dim(\bigcap_{i\in I} \Sb_i) +
    \dim(\Sb_j) - \dim(\bigcap_{i\in I} \Sb_i + \Sb_j)\notag\\
    \label{eq:7}
    &= p_J = p_I + \dim(\Sb_j) - \dim(\bigcap_{i\in I} \Sb_i + \Sb_j).
  \end{align}
  Necessity is established by substituting the expressions for $p_I$
  and $p_J$ into (\ref{eq:7}) and recalling that $\# J = \# I + 1$.

  ($\Leftarrow$) Let $K \subset I$. We have $\dim(\X) =
  \dim(\bigcap_{i\in I} \Sb_i + \Sb_j) \le \dim(\bigcap_{i\in K} \Sb_i
  + \Sb_j) \le \dim(\X)$. Taking $K=\{k\}$, jointly with the fact that
  $\{\Sb_i : i\in I\}$ is transverse, establishes that the dimension
  of the sum of the subspaces in every subset of $\{\Sb_i : i \in J\}$
  has maximum dimension. Also, we have
  \begin{align*}
    \dim(\bigcap_{i\in K} \Sb_i \cap \Sb_j) = \dim(\bigcap_{i\in K}
    \Sb_i) + \dim(\Sb_j) - \underbrace{\dim(\bigcap_{i\in K} \Sb_i + \Sb_j)}_{\dim(\X)},
  \end{align*}
  which, jointly with the fact that $\{\Sb_i : i\in I\}$ is
  transverse, establishes that the dimension of the intersection of
  the subspaces in every subset of $\{\Sb_i : i\in J\}$ has minimum
  dimension. 
\end{proof}
As is well known \cite{wonham_book85}, the property of transversality
is generic, i.e. it is satisfied for almost every set $S$ composed of
a finite number of subspaces of $\X$ selected ``randomly'' among all
subspaces of $\X$.

\subsection{Genericity of the SLASF Property}
\label{sec:induct-gener-struct}

As previously mentioned, we will look more deeply into the condition
(\ref{eq:35}). We define the following
\begin{equation}
  \label{eq:56}
  \B_i^\ell \dfn \img B_i^\ell,\qquad \B^\ell \dfn \bigcap_{i\in\Ind}
  \B_i^\ell.
\end{equation}
According to (\ref{eq:56}), then (\ref{eq:35}) can be rewritten as
$v_1^\ell \in \B^\ell$. Recall that $v_1^\ell$ is a
feedback-assignable common eigenvector. The following result is
straightforward.
\begin{Lemma}
  \label{lem:evimBsub}
  Let $\Sb_i^\ell$ be the set of vectors $v \in\B_i^\ell$ for which
  there exist $F_i^\ell$ and $\lambda$ with $|\lambda| < 1$ so that
  \begin{equation}
    \label{eq:1b}
    (A_i^\ell + B_i^\ell F_i^\ell)v = \lambda v.
  \end{equation}
  \begin{enumerate}[(a)]\itemsep-0pt
  \item The set $\Sb_i^\ell$ is a subspace.\label{item:10}
  \item $v\in\Sb_i^\ell$ if and only if $v\in\B_i^\ell$ and
    $A_i^\ell v \in\B_i^\ell$.\label{item:11}
  \end{enumerate}
\end{Lemma}
% \begin{proof}
%   TBW
% \end{proof}
By definition, $\Sb_i^\ell$ is the set of feedback-assignable
eigenvectors for the subsystem $(A_i^\ell, B_i^\ell)$ that are
contained in $\B_i^\ell$. Consequently, $v_1^\ell \in \B_i^\ell$ if
and only if $v_1^\ell \in \Sb_i^\ell$. In the sequel, we will employ
the following.
\begin{alignat}{2}
  \label{eq:2b}
  \rho_i^\ell &\dfn \dim(\Sb_i^\ell),&\qquad
  q_\ell &\dfn n_\ell + \sum_{i\in\Ind} \rho_i^\ell - Nn_\ell,\\
  \label{eq:3}
  \Sb^\ell &\dfn \bigcap_{i\in\Ind} \Sb_i^\ell, &\qquad \rho^\ell &\dfn
  \dim(\Sb^\ell).
\end{alignat}

Lemma~\ref{lem:dimcs} below relates the dimension of the subspace
$\Sb_i^\ell$ to the controllability indices of $(A_i^\ell,B_i^\ell)$.
\begin{Lemma}
  \label{lem:dimcs}
  The dimension $\dim(\Sb_i^\ell)=\rho_i^\ell$ equals the number of
  controllability indices equal to 1 of $(A_i^\ell,B_i^\ell)$.
\end{Lemma}
\begin{proof}
  According to the standard construction for the controllability
  indices of a system (see, e.g. \cite{wonham_book85}), it follows
  that the number of controllability indices equal to 1 of
  $(A_i^\ell,B_i^\ell)$ is given by $2m_i^\ell - \rank[\beta_i^\ell,
  A_i^\ell \beta_i^\ell]$, where $\beta_i^\ell$ is any matrix
  satisfying $\img \beta_i^\ell = \img B_i^\ell$. Since $\Sb_i^\ell
  \subset \B_i^\ell$, write $\B_i^\ell = \hat\B_i^\ell \oplus
  \Sb_i^\ell$ and let $\alpha = \dim(\hat\B_i^\ell)$. Then,
  $\rho_i^\ell = m_i^\ell - \alpha$. Let $\{b_1,\ldots,b_\alpha\}$ be
  a basis for $\hat\B_i^\ell$, $\{b_{\alpha+1},\ldots,b_{m_i^\ell}\}$
  be a basis for $\Sb_i^\ell$, and $\beta_i^\ell =
  [b_1,\ldots,b_{m_i^\ell}]$. By
  Lemma~\ref{lem:evimBsub}(\ref{item:11}), $A_i^\ell b_k \notin
  \B_i^\ell$ for $k=1,\ldots,\alpha$ and $A_i^\ell b_k \in \B_i^\ell$
  for $k=\alpha+1,\ldots,m_i^\ell$. Therefore, $\rank[\beta_i^\ell,
  A_i^\ell \beta_i^\ell] \le m_i^\ell + \alpha$. If
  $\rank[\beta_i^\ell, A_i^\ell \beta_i^\ell] < m_i^\ell + \alpha$,
  then $\sum_{j=1}^{m_i^\ell} c_j b_j + \sum_{k=1}^{\alpha} d_k
  A_i^\ell b_k = 0$ for some scalars $c_j$ and $d_k$, where not all
  the $d_k$ are zero. Then, $A_i^\ell \sum_{k=1}^{\alpha} d_k b_k \in
  \B_i^\ell$ and $A_i^\ell \sum_{k=1}^{\alpha} d_k b_k \notin
  \Sb_i^\ell$, a contradiction. Therefore, $\rank[\beta_i^\ell
  A_i^\ell \beta_i^\ell] = m_i^\ell + \alpha$ and $\rho_i^\ell =
  m_i^\ell - \alpha = 2m_i^\ell - \rank[\beta_i^\ell A_i^\ell
  \beta_i^\ell]$.
\end{proof}

Our main result is given below as Theorem~\ref{thm:genericity}. We
will provide comments and explanations after its proof. The proof of
Theorem~\ref{thm:genericity} requires an additional result, given as
Lemma~\ref{lem:cinditer}.
\begin{Theorem}
  \label{thm:genericity}
  Let 
  % $\{\B_i^1 : i\in\Ind\}$ be transverse,
  $\{\Sb_i^1 : i\in\Ind\}$ be transverse, %$p_1 > 0$
  and $q_1 \ge 0$. Then, $p_\ell > 0$ for $\ell=1,\ldots,n$.
\end{Theorem}
\begin{Lemma}
  \label{lem:cinditer}
  Consider Algorithm~\ref{alg:main} at iteration $\ell$. Suppose that
  $(A_i^\ell, B_i^\ell)$ is controllable and $A_i^{\ell,\CL} v_1^\ell
  = \lambda_i^\ell v_1^\ell$ with $v_1^\ell \neq 0$ and scalar
  $\lambda_i^\ell$. Then, $\Sb_i^{\ell+1} \supset U_{\ell+1}^*
  \Sb_i^\ell$, $(A_i^{\ell+1},B_i^{\ell+1})$ is controllable, and
  \begin{gather}
    \label{eq:42}
    \rho_i^{\ell+1} 
    \begin{cases}
      = \rho_i^\ell - 1 &\text{if } v_1^\ell\in\Sb_i^\ell,\\
      \ge \rho_i^\ell   &\text{otherwise.}
    \end{cases}% \\
  %   \label{eq:43}
  %   \Sb_i^{\ell+1} =
  %   \begin{cases}
  %     U_{\ell+1}^*\Sb_i^\ell &\text{if }i\in I,\\
  %     U_{\ell+1}^*\tilde\Sb_i^\ell &\text{if
  %     }i\in \bar I,
  %   \end{cases}
  \end{gather}
  % where $\Sb_i^\ell \subset \tilde\Sb_i^\ell$.
\end{Lemma}
\begin{proof}
  Let $\{t_j : j=1,\ldots,m_i^\ell \}$ be a basis for $\B_i^\ell$ and
  let $\kappa_{i,j}^\ell$, for $j=1,\ldots,m_i^\ell$ be the
  controllability indices of $(A_i^\ell, B_i^\ell)$. %  so that
  % $\{ (A_i^\ell)^k t_j : j=1,\ldots,m_i^\ell ;
  % k=0,\ldots,\kappa_{i,j}^\ell - 1\}$ is a basis for $\C^{n_\ell}$.
  By (\ref{eq:33}) and the feedback invariance of controllability
  indices, $\kappa_{i,j}^\ell$ also are the controllability indices of
  the pair $(A_i^{\ell,\CL},B_i^\ell)$. Since $(A_i^\ell,B_i^\ell)$ is
  controllable, then $(A_i^{\ell,\CL},B_i^\ell)$ also is controllable,
  and $D = \{ (A_i^{\ell,\CL})^k t_j : j=1,\ldots,m_i^\ell ;\:
  k=0,\ldots,\kappa_{i,j}^\ell - 1\}$ is a basis for
  $\R^{n_\ell}$. Write $v_1^\ell$ with respect to the basis $D$:
  $v_1^\ell = \sum_{j,k} c_{j,k} (A_i^{\ell,\CL})^k t_j$, where not
  all the $c_{j,k}$ are zero. Combining the latter with $A_i^{\ell,\CL} v_1^\ell
  = \lambda_i^\ell v_1^\ell$ yields
  \begin{equation}
    \label{eq:18}
    \sum_{j,k} c_{j,k} (A_i^{\ell,\CL})^{k+1} t_j 
    = \sum_{j,k} \lambda_i^\ell c_{j,k} (A_i^{\ell,\CL})^{k} t_j.
  \end{equation}
  From (\ref{eq:18}), it follows that $c_{j,k} \neq 0$ for at least
  one pair of indices $(j,k)$ such that $k=\kappa_{i,j}^\ell - 1$, or
  otherwise the vectors in $D$ would be linearly dependent, a
  contradiction. Let $\bar\kappa = \max_j \{\kappa_{i,j}^\ell :
  c_{j,k} \neq 0\text{ with }k=\kappa_{i,j}^\ell - 1\}$, and let
  $\bar\iota$ be such that $c_{\bar\iota,\bar\kappa - 1} \neq 0$. From
  the basis $D$, construct another basis, $\bar D$, by replacing the
  basis vector $(A_i^{\ell,\CL})^{\bar\kappa-1} t_{\bar\iota}$ by
  $v_1^\ell$. Note that $\spn\{U_{\ell+1}^* t : t\in\bar D\} =
  \R^{n_\ell-1}$ and $U_{\ell+1}^* v_1^\ell = 0$ (recall
  Section~\ref{sec:algorithm}). By (\ref{eq:62})--(\ref{eq:100}) and
  the fact that $A_i^{\ell,\CL} v_1^\ell = \lambda_i^\ell v_1^\ell$,
  then $U_{\ell+1}^* A_i^{\ell,\CL} = A_i^{\ell+1}
  U_{\ell+1}^*$. Hence $U_{\ell+1}^* (A_i^{\ell,\CL})^k t_j =
  (A_i^{\ell+1})^k U_{\ell+1}^* t_j$ and
  \begin{multline}
    \label{eq:19}
    \{(A_i^{\ell+1})^k U_{\ell+1}^* t_j : j=1,\ldots,m_i^\ell ;\\
    k=0,\ldots,\kappa_{i,j}^\ell - 1, (j,k)\neq (\bar\iota,\bar\kappa-1)\}
  \end{multline}
  is a basis for $\R^{n_\ell-1}$ [recall that, by (\ref{eq:5}),
  $n_{\ell+1} = n_\ell - 1$]. From (\ref{eq:6}), it follows that
  $\B_i^{\ell+1} = U_{\ell+1}^* \B_i^\ell$. We have that a basis
  for $\B_i^{\ell+1}$, is $E=\{U_{\ell+1}^* t_j : j=1,\ldots,m_i^\ell
  \}$ if $\bar\kappa > 1$ or $E=\{U_{\ell+1}^* t_j :
  j=1,\ldots,m_i^\ell; j\neq\bar\iota \}$ if $\bar\kappa = 1$. The
  condition $\bar\kappa = 1$ hence happens if and only if $v_1^\ell
  \in \Sb_i^\ell$. The preceding derivations show that the
  controllability indices of $(A_i^{\ell+1},B_i^{\ell+1})$ are given
  by $\kappa_{i,j}^{\ell+1} = \kappa_{i,j}^\ell$ for
  $j=1,\ldots,m_i^\ell$ with $j\neq\bar\iota$ and
  $\kappa_{i,\bar\iota}^{\ell+1} = \kappa_{i,\bar\iota}^\ell - 1$
  whenever $\bar\kappa = \kappa_{i,\bar\iota}^\ell > 1$. From the
  latter expressions, and recalling Lemma~\ref{lem:dimcs},
  (\ref{eq:42}) and the controllability of $(A_i^{\ell+1},
  B_i^{\ell+1})$ are established. Note that $\rho_i^{\ell+1} =
  \rho_i^\ell + 1$ whenever $\bar\kappa = 2$, since then
  $\kappa_{i,\bar\iota}^{\ell+1} = \bar\kappa - 1 = 1$ and hence
  $(A_i^{\ell+1},B_i^{\ell+1})$ has one controllability index equal to
  one more than $(A_i^\ell,B_i^\ell)$. The fact that $\Sb_i^{\ell+1}
  \supset U_{\ell+1}^* \Sb_i^\ell$ follows from the latter
  consideration and the basis $E$.
% Consider $(A_i^{\ell+1}, B_i^{\ell+1})$ as computed by
%   Algorithm~\ref{alg:main}. Then $\kappa_{i,j}^{\ell+1} =
%   \kappa_{i,j}^\ell$ for all $j=1,\ldots,m_i^\ell$ except for one, say
%   $j=q$. For $j=q$, we have $\kappa_{i,q}^{\ell+1} = \kappa_{i,q}^\ell
%   - 1$ if $\kappa_{i,q}^\ell > 1$. If $\kappa_{i,q}^\ell = 1$, then
%   $\dim(\img B_i^{\ell+1}) = m_i^\ell - 1$ and hence at iteration
%   $\ell+1$, the pair $(A_i^{\ell+1},B_i^{\ell+1})$ has one effective
%   input less than $(A_i^\ell,B_i^\ell)$. Moreover, if
%   $\kappa_{i,q}^\ell > 1$, then $\{
%   P(A_i^{\ell,\CL})^k t_j : j=1,\ldots,m_i^\ell ;
%   k=0,\ldots,\kappa_{i,j}^{\ell+1} - 1\}$ is a basis for
%   $\C^{n_{\ell+1}}$, and if $\kappa_{i,q}^\ell = 1$, then $\{
%   P(A_i^{\ell,\CL})^k t_j : j=1,\ldots,m_i^\ell, j\neq q ;
%   k=0,\ldots,\kappa_{i,j}^{\ell+1} - 1\}$ is a basis for $\C^{n_{\ell+1}}$.
%   % CLARIFY THAT ONE INPUT MAY DISAPPEAR !!!
\end{proof}

%We are now ready to prove Theorem~\ref{thm:genericity}.

\begin{proof}[Proof of Theorem~\ref{thm:genericity}]
  First, we prove that, if true, the following conditions
  \begin{equation}
    \label{eq:keystep}
    \begin{split}
      \{\Sb_i^\ell:i\in\Ind\} \text{ transverse},\quad
      q_\ell \ge0,\\
      (A_i^\ell,B_i^\ell) \text{ controllable},% &\quad p_\ell &> 0,
      % ,\quad p_\ell>0.
    \end{split}
  \end{equation}
  imply that $p_\ell > 0$. 
  % We next prove that $q_\ell \ge 0$ and $(A_i^\ell,B_i^\ell)$
  % controllable for all $i\in\Ind$ imply that $p_\ell >
  % 0$. 
  Since $\Sb_i^\ell \subset \B_i^\ell$, then $\rho_i^\ell \le
  m_i^\ell$ and $q_\ell \le p_\ell$. From controllability of
  $(A_i^\ell,B_i^\ell)$ and Lemma~\ref{lem:dimcs}, then $\rho_i^\ell =
  m_i^\ell$ if and only if $m_i^\ell = n_\ell$. Hence, if $q_\ell = p_\ell$,
  then $p_\ell = n_\ell > 0$. Otherwise, $0 \le q_\ell < p_\ell$.

  Next, we establish the validity of (\ref{eq:keystep}) for
  $\ell=1,\ldots,n$. Note that (\ref{eq:keystep}) hold at $\ell=1$ by
  assumption. Next, suppose that (\ref{eq:keystep}) hold at some $1
  \le \ell \le n-1$. By the argument in the previous paragraph, then
  $p_\ell > 0$, which ensures the existence and computation of
  $v_1^\ell \neq 0$ such that $A_i^{\ell,\CL} v_1^\ell =
  \lambda_i^\ell v_1^\ell$ with scalar $\lambda_i^\ell$ for all
  $i\in\Ind$. Hence, $(A_i^{\ell+1}, B_i^{\ell+1})$ is controllable by
  Lemma~\ref{lem:cinditer}. Also by Lemma~\ref{lem:cinditer}, we have
  $\Sb_i^{\ell+1} \supset U_{\ell+1}^* \Sb_i^\ell$ for all
  $i\in\Ind$. Since $q_\ell \ge 0$, by
  Lemma~\ref{lem:transverse}(\ref{item:7}) we have that $\rho^\ell =
  q_\ell$, and from Lemma~\ref{lem:transverse}(\ref{item:8}) we have
  $\dim(\Sb_i^\ell + \Sb_j^\ell) = n_\ell$ for all $i,j \in\Ind$ with
  $i\neq j$. It follows that $\dim(\Sb_i^{\ell+1} + \Sb_j^{\ell+1})
  \ge \dim(U_{\ell+1}^* (\Sb_i^\ell + \Sb_j^\ell)) = n_\ell -1 =
  n_{\ell+1}$ for all $i,j \in\Ind$ with $i\neq j$. 
  %%% BEGIN !!! CORRECTIONS !!!
%  \marginpar{\flushright{Corrections begin.}} %
  The latter fact establishes that the sum of the sets in every subset
  of $\{\Sb_i^{\ell+1}:i\in\Ind\}$ has maximal dimension and also that
  $\{\Sb_i^{\ell+1},\Sb_j^{\ell+1}\}$ is transverse for all
  $i,j\in\Ind$ with $i\neq j$. Let $T$ be a subset of
  $\{\Sb_i^{\ell+1}:i\in\Ind\}$. We proceed by induction on the number
  of subspaces in $T$. We have already established that $T$ is
  transverse if $\#T = 2$. Suppose next that $T$ is transverse
  whenever $\#T = 2,\ldots,\alpha$, with $\alpha \le N-1$. Let $T =
  \{\Sb_i^{\ell+1} : i\in I\}$, with $I\subset\Ind$ and $\#I =
  \alpha$, and consider $R = T \cup \{\Sb_j^{\ell+1}\}$ so that $\#R =
  \alpha+1$. By Lemma~\ref{lem:cinditer} and properties of maps and
  subspaces, we have
  \begin{align}
    \bigcap_{i\in I} \Sb_i^{\ell+1} + \Sb_j^{\ell+1} &\supset
    \bigcap_{i\in I} U_{\ell+1}^* \Sb_i^\ell + U_{\ell+1}^* \Sb_j^\ell\notag\\
    \label{eq:4}
    &\supset U_{\ell+1}^* (\bigcap_{i\in I} \Sb_i^\ell + \Sb_j^\ell).
  \end{align}
  By (\ref{eq:keystep}) and since $I \subset \Ind$, then $\{\Sb_i^\ell
  : i\in I\}$ is transverse. By
  Lemma~\ref{lem:transverse}(\ref{item:6}), then $\dim(\bigcap_{i\in
    I} \Sb_i^\ell + \Sb_j^\ell) = n_\ell$. Combining the latter
  equality with (\ref{eq:4}), then $\dim(\bigcap_{i\in I}
  \Sb_i^{\ell+1} + \Sb_j^{\ell+1}) = n_\ell - 1 = n_{\ell+1}$. By
  Lemma~\ref{lem:transverse}(\ref{item:6}) then $R$ is transverse. We
  have thus established that our induction hypothesis is valid for
  $\alpha+1$ and we conclude that $\{\Sb_i^{\ell+1}:i\in\Ind\}$ is
  transverse. %\marginpar{\flushright{Corrections End.}}
  %%%  END  !!! CORRECTIONS !!!
  By Lemma~\ref{lem:transverse}(\ref{item:7}), then $\rho^{\ell+1} =
  \max\{0,q_{\ell+1}\}$. From (\ref{eq:2b}) and (\ref{eq:42}), it
  follows that the minimum value for $q_{\ell+1}$ is $q_\ell - 1$, and
  this happens only if $\rho_i^{\ell+1} = \rho_i^\ell - 1$ for all
  $i\in\Ind$. However, if $q_\ell = 0$, then $\rho_i^{\ell+1} \ge
  \rho_i^\ell$ for at least one $i\in\Ind$ because, since $\rho^\ell =
  \dim(\Sb^\ell) = q_\ell = 0$, then $v_1^\ell \notin
  \Sb^\ell$. Consequently $q_{\ell+1} \ge 0$ and hence we have
  established (\ref{eq:keystep}) for $\ell=1,\ldots,n$.
\end{proof}
Theorem~\ref{thm:genericity} gives a condition, namely $\{\Sb_i^1 :
i\in\Ind\}$ transverse and $q_1 \ge 0$, under which the structural
condition of Lemma~\ref{cor:fail-gCEAS}(\ref{item:3}) is satisfied at
every iteration of Algorithm~\ref{alg:main}. The quantity $q_1$
depends on the dimensions of the subspaces $\Sb_i^1$ for $i\in\Ind$,
which, by Lemma~\ref{lem:dimcs} equals the number of controllability
indices equal to 1 of $(A_i^1 = A_i, B_i^1 = B_i)$. A consequence of
Corollary 5.4 of \cite{wonham_book85} is that the latter number is
generically (i.e. for almost every matrices $A_i$ and $B_i$) nonzero
when $m_i > n/2$ and if the latter holds, generically equal to $m_i -
(n \mod m_i)$. If $m_i > n/2$, arbitrary choices for the entries of
$A_i$ and $B_i$ yield arbitrary $\Sb_i^1$, although generically of
dimension $m_i - (n \mod m_i)$. Therefore, if the system dimensions
$n$ and $m_i$ for $i\in\Ind$ are such that $q_1 \ge 0$, then
$\{\Sb_i^1 : i\in\Ind\}$ will be transverse generically in the space
of parameters of the matrices $A_i$ and $B_i$, for $i\in\Ind$. For
example, a DTSS with two subsystems ($N=2$), order $n=6$, subsystem 1
having 4 inputs ($m_1 = 4$) and subsystem 2 having 5 inputs ($m_2 =
5$) will generically satisfy the hypotheses of
Theorem~\ref{thm:genericity}, since $\rho_1^1 = m_1 - (n \mod m_1) =
2$, $\rho_2^1 = 4$ and hence $q_1 = 6 + (2+4) - 2\cdot 6 = 0$. From the
preceding analysis, it follows that the hypotheses of
Theorem~\ref{thm:genericity} can hold only for DTSSs where each
subsystem has ``a lot of'' inputs ($m_i > n/2$, $q_1 \ge 0$).

\section{CONCLUSIONS}
\label{sec:conclusions}

We have addressed feedback stabilisation of discrete-time switched
linear systems with control inputs. The control strategy employed is
to seek feedback matrices so that the closed-loop subsystem matrices
are stable and generate a solvable Lie algebra. The problem of
feedback stabilisation by means of the latter strategy is known to not
always have a solution. In this context, we have derived conditions
under which this problem is ensured to have a solution for most
possible sets of parameters. These conditions hold only for DTSSs
of specific system dimensions, where each subsystem has a considerable
number of inputs, as compared with the system dimension.

\bibliographystyle{plain}

{\small\bibliography{cdc2009}}

\end{document}